\documentclass[a4paper,onecolumn,11pt,accepted=2025-07-11]{quantumarticle}

\pdfoutput=1
\usepackage[utf8]{inputenc}
\usepackage[english]{babel}
\usepackage[T1]{fontenc}
\usepackage{amsmath}
\usepackage{hyperref}

\usepackage{tikz}
\usepackage{lipsum}

\usepackage{amsfonts}
\usepackage{amsthm}
\theoremstyle{plain}

\newtheorem{proposition}{Proposition}
\newtheorem{lemma}{Lemma}
\newtheorem{corollary}{Corollary}
\newtheorem{remark}{Remark}
\usepackage{bm}
\usepackage{braket}
\usepackage{float} 
\usepackage[numbers]{natbib}
\begin{document}

\title{Estimation of Quantum Fisher Information via Stein’s Identity in Variational Quantum Algorithms}

\author{Mourad Halla}
\affiliation{Deutsches Elektronen-Synchrotron DESY, Platanenallee 6, 15738 Zeuthen, Germany}

\maketitle

\begin{abstract}
The Quantum Fisher Information Matrix (QFIM) plays a crucial role in quantum optimization algorithms such as Variational Quantum Imaginary Time Evolution and Quantum Natural Gradient Descent. However, computing the full QFIM incurs a quadratic computational cost of \( O(d^2) \) with respect to the number of parameters \( d \), limiting its scalability for high-dimensional quantum systems. To address this limitation, stochastic methods such as the Simultaneous Perturbation Stochastic Approximation (SPSA) have been employed to reduce computational complexity to a constant [Quantum \textbf{5}, 567 (2021)]. In this work, we propose an alternative estimation framework based on Stein’s identity that also achieves constant computational complexity. Furthermore, our method reduces the quantum resources required for QFIM estimation compared to the SPSA approach. We provide numerical examples using the transverse-field Ising model and the lattice Schwinger model to demonstrate the feasibility of applying our method to realistic quantum systems.
\end{abstract}
\section{Introduction}

Quantum computing has emerged as a transformative framework for tackling computational challenges that exceed the reach of classical algorithms. Among the leading strategies in this field, \textit{Variational Quantum Algorithms (VQAs)} \cite{Cerezo2021, McClean2016, Bharti2022}, which use a hybrid quantum-classical optimization scheme, have garnered significant attention due to their applicability in areas such as quantum chemistry, materials science, and high-energy physics. A particularly impactful algorithm within this class is the \textit{Variational Quantum Eigensolver (VQE)} \cite{Peruzzo2014}, which is designed to approximate ground-state energies of quantum systems, making it especially well-suited for near-term noisy intermediate-scale quantum (NISQ) devices.

Optimization strategies play a pivotal role in the efficiency of VQAs, influencing the quality of the approximated solutions and the quantum resources required to reach them. Conventional optimization algorithms, such as gradient descent (GD), often struggle to navigate the intricate nature of quantum optimization landscapes, as these landscapes in VQAs are typically non-convex and prone to noise and barren plateaus. This challenge necessitates the development of more sophisticated optimization techniques to enhance convergence speed and improve the accuracy of approximated solutions. Some efficient quantum optimization algorithms include Quantum Natural Gradient (QNG) \cite{Stokes2020} and Variational Quantum Imaginary Time Evolution (VarQITE) \cite{McArdle} that use an update rule based on the Quantum Fisher Information Matrix (QFIM).

The QFIM serves as a fundamental metric that characterizes the local curvature of a quantum state manifold, providing a natural geometric framework for parameter optimization. Unlike classical approaches that rely on Euclidean distances in parameter space, the QFIM captures how quantum states evolve under small parameter variations, ensuring an optimization trajectory aligned with the true structure of quantum state space. Moreover, QFIM-based techniques have broader implications beyond optimization in VQAs, contributing to quantum sensing and quantum metrology (see the review in~\cite{Meyer2021, Liu2020, S.Sidhu2020}).

Despite this, standard methods like the parameter-shift rule \cite{Mari} for computing the QFIM scale quadratically as \( O(d^2) \) with the number of parameters \( d \), making them impractical for high-dimensional quantum systems. This limitation necessitates alternative approaches for efficiently computing the QFIM. In \cite{Stokes2020}, diagonal and block-diagonal approximation were proposed to reduce computational complexity, but these methods result in the loss of parameter correlations in the QFIM. More efficient approaches that preserve these correlations involve the use of the simultaneous perturbation stochastic approximation (SPSA) \cite{Spall1992} method to approximate the QFIM, reducing the complexity cost to a constant \cite{Gacon}. In this work, we propose an alternative and efficient method based on Stein’s identity, which also reduces the quantum computational complexity from \( O(d^2) \) to a constant without losing parameter correlations.

Stein’s identity provides a powerful framework for estimating Hessian information in stochastic optimization, particularly in zeroth-order (ZO) settings where direct gradient and Hessian computations are infeasible. Recent studies have leveraged it to develop more efficient Hessian approximation methods, reducing per-iteration complexity compared to traditional second-order techniques such as the simultaneous perturbation stochastic approximation (2SPSA) \cite{Zhu2022,Erdogdu2015}. Employing a perturbation-based approach enables unbiased gradient and Hessian estimation while requiring fewer ZO queries than 2SPSA, thereby improving convergence efficiency. In this work, we use a similar idea to estimate the QFIM.

The remainder of the manuscript is organized as follows: Section~\ref{sec:Theory} provides a mathematical overview of Stein’s identity, demonstrating its application in estimating gradients, the Hessian, and the QFIM. Section~\ref{sec:QNG} presents a brief overview of variational quantum algorithms and the Quantum Natural Gradient. Section~\ref{sec:numerical} includes numerical examples using VQE to compute the ground-state energy of the transverse-field Ising model and the lattice Schwinger model. Finally, Section~\ref{sec:conclusion} provides a summary and outlook.
\section{Theoretical Results}
\label{sec:Theory}
In this section, we provide a theoretical overview of gradient and Hessian estimation using Stein’s identity, and subsequently integrate the Stein’s identity method into the QFIM framework. As Stein’s identity has been proposed as an alternative to Simultaneous Perturbation Stochastic Approximation (SPSA) for reducing computational complexity, we first review the SPSA method before introducing the Stein-based approach.

\subsection{Gradient and Hessian Estimation via SPSA Methods}

Simultaneous Perturbation Stochastic Approximation (SPSA) \cite{Spall1992} provides an efficient method for estimating gradients using stochastic perturbations, significantly reducing computational costs compared to finite-difference approaches. Given a differentiable objective function \(f\colon \mathbb{R}^d \to \mathbb{R}\), SPSA estimates the gradient using a two-point perturbation method.

Let \(\bm{\Delta}\sim\{\pm1\}^d\) be a random perturbation vector drawn from the Rademacher distribution (each coordinate independently takes \(\pm1\) with probability~\(1/2\)). Then the gradient is estimated as
\begin{equation}
  \nabla f(\bm{\theta})
  =
  \mathbb{E}_{\bm{\Delta}}\!\Bigl[
    \frac{f(\bm{\theta} + c\,\bm{\Delta}) - f(\bm{\theta} - c\,\bm{\Delta})}{2c}\,\bm{\Delta}
  \Bigr],
\end{equation}
where \(c>0\) is a small finite displacement. This method requires only two function evaluations per iteration, and \(\bm{\Delta}\) is independent of \(\bm{\theta}\), making it scalable to high-dimensional problems.

The second-order SPSA (2SPSA) extends SPSA to approximate the Hessian matrix \(\bm{H}(\bm{\theta})\) without explicitly computing all \(d^2\) entries. Two independent perturbation vectors \(\bm{\Delta}_1\) and \(\bm{\Delta}_2 \sim\{\pm1\}^d\) are sampled, leading to the Hessian estimate
\begin{equation}
  \hat{\bm{H}}
  =
  \mathbb{E}_{\bm{\Delta}_1,\bm{\Delta}_2}\!\Bigl[
    \frac{\delta f}{2c^2}\,
    \frac{\bm{\Delta}_1\,\bm{\Delta}_2^\mathsf{T} + \bm{\Delta}_2\,\bm{\Delta}_1^\mathsf{T}}{2}
  \Bigr],
  \label{Hessian_2SPSA}
\end{equation}
where
\begin{equation}
  \delta f
  =
  f(\bm{\theta} + c\,\bm{\Delta}_1 + c\,\bm{\Delta}_2)
  - f(\bm{\theta} + c\,\bm{\Delta}_1)
  - f(\bm{\theta} - c\,\bm{\Delta}_1 + c\,\bm{\Delta}_2)
  + f(\bm{\theta} - c\,\bm{\Delta}_1).
\end{equation}
This approach reduces the \(O(d^2)\) computational cost of explicit Hessian estimation while maintaining reasonable accuracy. The 2SPSA method requires four function evaluations and two perturbation vectors \(\bm{\Delta}_1\) and \(\bm{\Delta}_2\). In the following sections, we will see that Stein’s identity requires two or three function evaluations and a single perturbation vector.
\subsection{Gradient and Hessian Estimation via Stein’s Identity}
\label{sec:Stein}

Stein’s identity provides a fundamental mathematical tool for estimating gradients and Hessians, particularly in optimization problems where derivative information is inaccessible or computationally expensive. In this section, we provide an overview of this method and its application in estimating gradients and Hessians. For a more detailed mathematical treatment and rigorous proofs, we refer the reader to \cite{Stein2004, Zhu2022, Erdogdu2015}.

A central result of Stein’s method is summarized in the following proposition:
\begin{proposition}[First-Order and Second-Order Stein’s Identity \cite{Stein2004}] 
Let \( \bm{X} \in \mathbb{R}^d \) represent a \(d\)-dimensional random vector with an underlying probability density function \( p(\bm{x}): \mathbb{R}^d \to \mathbb{R} \).
\begin{itemize}
    \item[i)] If \( p(\bm{x}) \) is differentiable, and \( q: \mathbb{R}^d \to \mathbb{R} \) is a differentiable function such that \( \mathbb{E} \left\{ \bm{\nabla} q(\bm{X}) \right\} \) exists, then the following identity holds:
    \begin{equation}
        \mathbb{E} \left\{ q(\bm{X})[p(\bm{X})]^{-1} \bm{\nabla} p(\bm{X}) \right\} = -\mathbb{E} \left\{ \bm{\nabla} q(\bm{X}) \right\}.
        \label{eq:first_order_stein}
    \end{equation}
    \item[ii)] If \( p(\bm{x}) \) and \( q(\bm{x}) \) are twice differentiable functions such that \( \mathbb{E} \left\{ \bm{\nabla}^2 q(\bm{X}) \right\} \) exists, then:
    \begin{equation}
        \mathbb{E} \left\{ q(\bm{X})[p(\bm{X})]^{-1} \bm{\nabla}^2 p(\bm{X}) \right\} = \mathbb{E} \left\{ \bm{\nabla}^2 q(\bm{X}) \right\}.
        \label{eq:second_order_stein}
    \end{equation}
\end{itemize}
\label{proposition1}
\end{proposition}
When considering a multivariate standard normal random vector \( \bm{X} \sim \mathcal{N}(\bm{0}, \bm{I}) \), where \( \bm{I} \) represents the identity matrix, we have \( \bm{\nabla} p(\bm{x}) = -\bm{x}p(\bm{x}) \) and \( \bm{\nabla}^2 p(\bm{x}) = (\bm{x}\bm{x}^T - \bm{I})p(\bm{x}) \). Using these expressions, equations \eqref{eq:first_order_stein} and \eqref{eq:second_order_stein} in Proposition \ref{proposition1} take the form:
\begin{equation}
    \mathbb{E} \left\{ \bm{X} q(\bm{X}) \right\} = \mathbb{E} \left\{ \bm{\nabla} q(\bm{X}) \right\},
    \label{eq:normal_first_order}
\end{equation}
\begin{equation}
    \mathbb{E} \left\{ (\bm{X} \bm{X}^T - \bm{I}) q(\bm{X}) \right\} = \mathbb{E} \left\{ \bm{\nabla}^2 q(\bm{X}) \right\}.
    \label{eq:normal_second_order}
\end{equation}
For the case of \( \bm{X} \sim \mathcal{N}(\bm{0}, \bm{\Sigma}) \) where \( \bm{\Sigma} \) is an arbitrary positive definite covariance matrix, the gradient and Hessian of the probability density function are given by \( \bm{\nabla} p(\bm{x}) = -\bm{\Sigma}^{-1} \bm{x} p(\bm{x}) \) and \( \bm{\nabla}^2 p(\bm{x}) = (\bm{\Sigma}^{-1} \bm{x} \bm{x}^T \bm{\Sigma}^{-1} - \bm{\Sigma}^{-1}) p(\bm{x}) \). Under these conditions, equations \eqref{eq:first_order_stein} and \eqref{eq:second_order_stein} reduce to:
\begin{equation}
    \mathbb{E} \left\{ \bm{\Sigma}^{-1} \bm{X} q(\bm{X}) \right\} = \mathbb{E} \left\{ \bm{\nabla} q(\bm{X}) \right\},
    \label{eq:first_order_normal}
\end{equation}
\begin{equation}
    \mathbb{E} \left\{ (\bm{\Sigma}^{-1} \bm{X} \bm{X}^T \bm{\Sigma}^{-1} - \bm{\Sigma}^{-1}) q(\bm{X}) \right\} = \mathbb{E} \left\{ \bm{\nabla}^2 q(\bm{X}) \right\}.
    \label{eq:second_order_normal}
\end{equation}

To approximate the gradient and Hessian of a function \(f(\bm{\theta})\) in optimization problems, we define, for any \(c>0\) and \(\bm{u} \sim \mathcal{N}(\bm{0}, \bm{I})\),
\begin{equation}
  f_{c}(\bm{\theta}) \;=\; \mathbb{E}_{\bm{u}}\bigl[f(\bm{\theta} + c\,\bm{u})\bigr].
  \label{fc}
\end{equation}
The function \(f_{c}(\bm{\theta})\) is thus \(C^{\infty}\). Stein’s identity offers various approximation orders, requiring one, two, or three function evaluations to estimate either the unbiased gradient or the unbiased Hessian. The results are summarized in the following lemma:

\begin{lemma}[Stein’s Identity-Based Estimator \cite{Zhu2022}]
Let \(\bm{u} \sim \mathcal{N}(\bm{0}, \bm{I})\), and consider \(f_{c}(\bm{\theta})\) defined by \eqref{fc}. Then the gradient \(\bm{\nabla} f_{c}(\bm{\theta})\) and Hessian \(\bm{\nabla}^{2} f_{c}(\bm{\theta})\) can be estimated as follows:
\begin{itemize}
    \item[i)] \textbf{Single-Evaluation Estimator:}
    \begin{equation}
        \bm{\nabla} f_{c}(\bm{\theta})
        =
        \mathbb{E}_{\bm{u}}\bigl\{\,c^{-1}\,f(\bm{\theta} + c\,\bm{u})\,\bm{u}\bigr\},
        \label{eq:single_evaluation_gradient}
    \end{equation}
    \begin{equation}
        \bm{\nabla}^{2} f_{c}(\bm{\theta})
        =
        \mathbb{E}_{\bm{u}}\bigl\{\,c^{-2}\,f(\bm{\theta} + c\,\bm{u})\,(\bm{u}\,\bm{u}^{T} - \bm{I})\bigr\}.
        \label{eq:single_evaluation_hessian}
    \end{equation}

    \item[ii)] \textbf{Two-Evaluation Estimator:}
    \begin{equation}
        \bm{\nabla} f_{c}(\bm{\theta})
        =
        \mathbb{E}_{\bm{u}}\bigl\{\,(2c)^{-1}\,\bigl(f(\bm{\theta} + c\,\bm{u}) - f(\bm{\theta} - c\,\bm{u})\bigr)\,\bm{u}\bigr\},
        \label{eq:two_evaluation_gradient}
    \end{equation}
    \begin{equation}
        \bm{\nabla}^{2} f_{c}(\bm{\theta})
        =
        \mathbb{E}_{\bm{u}}\bigl\{\,c^{-2}\,\bigl(f(\bm{\theta} + c\,\bm{u}) - f(\bm{\theta})\bigr)\,(\bm{u}\,\bm{u}^{T} - \bm{I})\bigr\}.
        \label{eq:two_evaluation_hessian}
    \end{equation}

    \item[iii)] \textbf{Three-Evaluation Estimator:}
    \begin{equation}
        \bm{\nabla}^{2} f_{c}(\bm{\theta})
        =
        \mathbb{E}_{\bm{u}}\bigl\{\, (2\,c^{2})^{-1}\,\bigl(f(\bm{\theta} + c\,\bm{u}) + f(\bm{\theta} - c\,\bm{u}) - 2\,f(\bm{\theta})\bigr)\,(\bm{u}\,\bm{u}^{T} - \bm{I})\bigr\}.
        \label{eq:three_evaluation_hessian}
    \end{equation}
\end{itemize}
\label{Lemma1}
\end{lemma}
To provide flexibility for QFIM estimation in the next section, we extend Lemma~\ref{Lemma1} to the case where \(\bm{X} \sim \mathcal{N}(\bm{0},\,\bm{\Sigma})\), with an arbitrary positive-definite covariance matrix \(\bm{\Sigma}\). In this work, we focus on the special case where \(\bm{\Sigma} = b^{2}\,\bm{I}\), and summarize our results for Hessian estimation in the following corollary:

\begin{corollary}
Let \(\bm{X} \sim \mathcal{N}(\bm{0},\,b^{2}\,\bm{I})\), and define \(\bm{X} = c\,\bm{Y}\), with \(\bm{Y} \sim \mathcal{N}(\bm{0},\,\tfrac{b^{2}}{c^{2}}\,\bm{I})\). Assume \(b,c>0\). Then the two-evaluation estimator of the Hessian of the smoothed function \(f_{c}(\bm{\theta})\) is given by:
\begin{equation}
    \bm{\nabla}^{2} f_{c}(\bm{\theta}) 
    = 
    \mathbb{E}_{\bm{Y}} \Bigl\{\tfrac{c^{2}}{b^{4}}\,\bigl[f(\bm{\theta} + c\,\bm{Y}) - f(\bm{\theta})\bigr]\,\bigl(\bm{Y}\,\bm{Y}^{T} - \tfrac{b^{2}}{c^{2}}\,\bm{I}\bigr)\Bigr\}.
    \label{eq:hessian_general_cov_2}
\end{equation}
and the three-evaluation estimator is given by:
\begin{equation}
    \bm{\nabla}^{2} f_{c}(\bm{\theta}) 
    = 
    \mathbb{E}_{\bm{Y}} \Bigl\{\tfrac{c^{2}}{2\,b^{4}}\,\bigl[f(\bm{\theta} + c\,\bm{Y}) + f(\bm{\theta} - c\,\bm{Y}) - 2\,f(\bm{\theta})\bigr]\,\bigl(\bm{Y}\,\bm{Y}^{T} - \tfrac{b^{2}}{c^{2}}\,\bm{I}\bigr)\Bigr\}.
    \label{eq:hessian_general_cov_3}
\end{equation}
\end{corollary}

\begin{proof}
The result follows from Stein’s second-order identity \eqref{eq:second_order_normal}, applied to the transformation \(\bm{X} = c\,\bm{Y}\). Given that \(\bm{X} \sim \mathcal{N}(\bm{0},\,b^{2}\,\bm{I})\), substituting \(\bm{X} = c\,\bm{Y}\) ensures that \(\bm{Y} \sim \mathcal{N}(\bm{0},\,\tfrac{b^{2}}{c^{2}}\,\bm{I})\).
\end{proof}

In the next section, we extend a similar approach used for Hessian estimation to the Quantum Fisher Information Matrix.
\subsection{QFIM Estimation via Stein’s Identity}
\label{subsec:QFIM-Stein}

Similar to the approach proposed by Gacon~\cite{Gacon}, which employs 2SPSA~\cite{Spall1992} to estimate the QFIM, we propose an alternative method based on Stein’s identity for QFIM estimation. We begin with a brief discussion of the QFIM and refer to~\cite{Meyer2021, Liu2020, S.Sidhu2020} for a general overview and its role in variational quantum algorithms.

Let \(\ket{\psi(\bm{\theta})}\) be a parameterized pure quantum state in an \(n\)-qubit Hilbert space, where \(\bm{\theta} \in \mathbb{R}^d\) represents a set of \(d\) trainable parameters. Fix a small displacement \(\delta\bm{\theta} \in \mathbb{R}^d\) and define the overlap function:
\begin{equation}
  \mathcal{F}(\bm{\theta})
  :=
  \bigl|\langle \psi(\bm{\theta}) \mid \psi(\bm{\theta} + \delta\bm{\theta})\rangle\bigr|^{2}.
  \label{eq:one-arg-fidelity}
\end{equation}
The Fubini–Study metric tensor is given by
\begin{equation}
  F_{ij}(\bm{\theta})
  =
  -\frac{1}{2}\,
  \partial_{i}\partial_{j}
  \,\mathcal{F}(\bm{\theta})\Big|_{\delta\bm{\theta}=0},
  \label{QFI-Hess-Like}
\end{equation}
so that the QFIM equals \(4\,F_{ij}(\bm{\theta})\). Explicitly,
\begin{equation}
  F_{ij}(\bm{\theta})
  =
  \mathrm{Re}\!\Bigl[
    \langle \partial_{i}\psi(\bm{\theta}) \mid \partial_{j}\psi(\bm{\theta})\rangle
    -
    \langle \partial_{i}\psi(\bm{\theta}) \mid \psi(\bm{\theta})\rangle\,
    \langle \psi(\bm{\theta}) \mid \partial_{j}\psi(\bm{\theta})\rangle
  \Bigr],
  \label{eq:Full-metric}
\end{equation}
where \(\partial_i\) and \(\partial_j\) denote partial derivatives with respect to \(\theta_i\) and \(\theta_j\). See Appendix~\ref{app:Fubini} for detailed computations. Evaluating \eqref{eq:Full-metric} using methods such as the parameter-shift rule requires \(O(d^2)\) computations, rendering it impractical for large-scale VQA optimization.

We now examine how this issue can be addressed using Stein’s identity. At first glance, we observe that \eqref{QFI-Hess-Like} resembles a Hessian, which allows us to apply the results from Section~\ref{sec:Stein} for its estimation. For practical purposes, we focus on the two-evaluation and three-evaluation methods, presented in \eqref{eq:hessian_general_cov_2} and \eqref{eq:hessian_general_cov_3}, respectively. Furthermore, assuming that the rows of \(\bm{Y}\) are independent and identically distributed (i.i.d.) random vectors drawn from \(\mathcal{N}\bigl(\bm{0},\,\tfrac{b^2}{c^2}\bm{I}\bigr)\). Higher accuracy can be achieved by generating \(N\) independent perturbation vectors \(\bm{Y}_{i}\).

\medskip
Let \(c > 0\) and \(\bm{Y} \sim \mathcal{N}\bigl(\bm{0},\,\tfrac{b^2}{c^2}\bm{I}\bigr)\). Based on \eqref{eq:one-arg-fidelity}, define the smoothed function
\begin{equation}
  \mathcal{F}_{c}(\bm{\theta})
  :=
  \mathbb{E}_{\bm{Y}}
  \bigl[\mathcal{F}(\bm{\theta} + c\,\bm{Y})\bigr].
\end{equation}
The two-evaluation Stein estimator for the metric tensor is
\begin{equation}
    \hat{\bm{F}} 
    =
    -\frac{c^2}{2\,b^4\,N}
    \sum_{i=1}^N
      \bigl[\mathcal{F}(\bm{\theta} + c\,\bm{Y}_i) - \mathcal{F}(\bm{\theta})\bigr]
      \bigl(\bm{Y}_i\,\bm{Y}_i^\top - \tfrac{b^2}{c^2}\,\bm{I}\bigr),
    \label{Stein_F2}
\end{equation}
and the three-evaluation estimator is
\begin{equation}
    \hat{\bm{F}} 
    =
    -\frac{c^2}{4\,b^4\,N}
    \sum_{i=1}^N
      \bigl[\mathcal{F}(\bm{\theta} + c\,\bm{Y}_i) + \mathcal{F}(\bm{\theta} - c\,\bm{Y}_i) - 2\,\mathcal{F}(\bm{\theta})\bigr]
      \bigl(\bm{Y}_i\,\bm{Y}_i^\top - \tfrac{b^2}{c^2}\,\bm{I}\bigr).
    \label{Stein_F3}
\end{equation}
Both estimators \eqref{Stein_F2} and \eqref{Stein_F3} are unbiased for \(\bm{\nabla}^{2}\mathcal{F}_{c}(\bm{\theta})\), i.e., \(\mathbb{E}[\hat{\bm{F}}] = \bm{\nabla}^{2}\mathcal{F}_{c}(\bm{\theta})\). We emphasize that by letting the smoothing parameter \(c \to 0\) while independently increasing the Monte Carlo sample size \(N \to \infty\), the bias of \eqref{Stein_F2}–\eqref{Stein_F3} with respect to the true metric tensor \(\bm{\nabla}^{2}\mathcal{F}(\bm{\theta})\) vanishes asymptotically as \(O(c^2)\). In practice, one chooses \(c\) and \(N\) so that this residual bias is negligible compared to statistical fluctuations or shot noise.

In equations \eqref{Stein_F2} and \eqref{Stein_F3}, the computational complexity is reduced from \( O(d^2) \) to a constant, making it independent of the number of parameters \( d \). Compared to the 2SPSA QFIM estimator proposed by Gacon~\cite{Gacon}, which requires four circuit evaluations, our method based on Stein’s identity offers greater flexibility. Users can choose between the two-circuit evaluation QFIM estimator, \eqref{Stein_F2}, or the three-circuit evaluation QFIM estimator, \eqref{Stein_F3}. Furthermore, the 2SPSA QFIM method requires two independent perturbation vectors, \(\bm{\Delta}_1\) and \(\bm{\Delta}_2\), whereas the Stein-based method requires only a single perturbation vector, \(\bm{Y}\). The approximation error in the QFIM estimates scales as \( O(N^{-1/2}) \).

The Fubini–Study metric, or QFIM, involves evaluating the squared overlap between quantum states \( |\psi(\bm{\theta})\rangle \) and \( |\psi(\bm{\theta}+ \delta\bm{\theta})\rangle \), i.e., \(\left| \langle \psi(\bm{\theta}) \mid \psi(\bm{\theta}+ \delta\bm{\theta}) \rangle \right|^2\). The estimation method employed in this work prepares the state \( U^\dagger(\bm{\theta}+ \delta\bm{\theta}) U(\bm{\theta}) |0\rangle \), where \( U(\bm{\theta}) \) is a parameterized unitary, and measures the probability of \( |0\rangle \), which directly yields the overlap. This approach maintains the original circuit width \( n \), where \( n \) is the number of qubits, but doubles the circuit depth to \( 2m \), with \( m \) representing the depth of \( U(\bm{\theta}) \).

In this manuscript, we utilize \eqref{Stein_F2} and \eqref{Stein_F3} as examples within the framework of the Quantum Natural Gradient (QNG) applied to the Variational Quantum Eigensolver (VQE). We leave the application to imaginary time evolution for future work. In the following section, we provide an overview of QNG.
\section{Quantum Natural Gradient and VQE}
\label{sec:QNG}

In the Variational Quantum Eigensolver (VQE), the objective function is typically given by
\begin{equation}
    \mathcal{L}(\bm{\theta}) = \bra{0} U^\dagger(\bm{\theta}) O U(\bm{\theta}) \ket{0},
\end{equation}
where \( O \) is a Hermitian operator, and the parameterized unitary is
\begin{equation}
    U(\bm{\theta}) = \prod_{\ell=1}^{p} W_{\ell} \exp(i\theta_{\ell} X_{\ell}), 
\end{equation}
with \( W_{\ell} \) and \( X_{\ell} \) being fixed unitary and Hermitian operators, respectively.

The parameter vector \( \bm{\theta} \) is iteratively updated to minimize the objective function \( \mathcal{L}(\bm{\theta}) \). In standard gradient descent, the update rule is given by
\begin{equation}
    \bm{\theta}_{k+1} = \bm{\theta}_k - \eta \bm{\nabla} \mathcal{L}(\bm{\theta}_k),
\end{equation}
where \( \eta > 0 \) is a user-defined learning rate (step size), and \( \bm{\nabla} \mathcal{L}(\bm{\theta}_k) \) denotes the Euclidean gradient.

In the Quantum Natural Gradient (QNG) method, the Euclidean gradient is replaced by the Riemannian gradient \( \bm{F}^{-1}(\bm{\theta}_k) \bm{\nabla} \mathcal{L}(\bm{\theta}_k) \), where \( \bm{F}^{-1}(\bm{\theta}_k) \) is the inverse of the Fubini–Study metric tensor. The update rule then becomes
\begin{equation}
    \bm{\theta}_{k+1} = \bm{\theta}_k - \eta \bm{F}^{-1}(\bm{\theta}_k) \bm{\nabla} \mathcal{L}(\bm{\theta}_k).
    \label{eq:QNG_updat}
\end{equation}
To incorporate previous stochastic estimates of the metric tensor, we replace \( \bm{F} \) in the update rule with \( \bar{\bm{F}}_k \):
\begin{equation}
    \bar{\bm{F}}_k = \frac{k}{k+1} \bar{\bm{F}}_{k-1} + \frac{1}{k+1} \hat{\bm{F}}_{k},
\end{equation}
where the metric estimate \( \bar{\bm{F}}_k \) incorporates all previous samples. This recursive averaging introduces a small geometric error, as it combines metric tensors evaluated at slightly different parameter points, \( \bm{\theta}_{k-1} \) and \( \bm{\theta}_k \). However, for sufficiently small step sizes, this error remains negligible in practice (see Remark~\ref{rem_tensor} in Appendix~\ref{Tensor_error}).

To ensure the metric remains positive semi-definite and invertible, we apply the regularization strategy described in \cite{Gacon}. In particular, the metric defined in equation~\eqref{Stein_F2} or~\eqref{Stein_F3} is replaced by
\begin{equation}
    \hat{\bm{F}}_{\mathrm{reg}} = \frac{\sqrt{\hat{\bm{F}}^{\!\top} \hat{\bm{F}}} + \beta\,\bm{I}}{1 + \beta},
    \label{reg}
\end{equation}
where \( \beta \) is a tunable regularization coefficient selected by empirical testing. If \( \beta \) is too small, the adjusted metric may fail to remain positive semi-definite; if \( \beta \) is too large, the contribution of the metric is effectively erased, reducing the quantum natural gradient update to a first-order update and canceling the curvature information carried by the metric.

Unlike the Moore–Penrose pseudo-inverse, the regularization in equation~\eqref{reg} avoids uncontrolled amplification of noise-dominated eigenvalues and prevents discontinuous jumps in the inverse metric when small eigenvalues cross zero, thereby preserving a smooth, positive-definite Riemannian metric that yields stable natural-gradient updates.

In addition to the regularization of the metric, a blocking mechanism may be enforced by setting \( \bm{\theta}_{k+1} = \bm{\theta}_k \) if the evaluation of the noisy objective function at \( \bm{\theta}_{k+1} \) is substantially higher than at \( \bm{\theta}_k \) by a user-specified constant. This ensures stability in the update rule and prevents divergence.

For the simulations in the next section, we use in the update rule \eqref{eq:QNG_updat} the metric estimators \eqref{Stein_F2} and \eqref{Stein_F3}, regularized according to \eqref{reg}. For the gradient, we employ the two-function evaluation estimator given in equation~\eqref{eq:two_evaluation_gradient}:
\begin{equation}
    \hat{\bm{g}}_k = \frac{1}{N} \sum_{i=1}^N (2c)^{-1} \big(\mathcal{L}(\bm{\theta}_k + c \bm{u}_{k,i}) - \mathcal{L}(\bm{\theta}_k - c \bm{u}_{k,i})\big) \bm{u}_{k,i},
    \label{Gradient_Stein}
\end{equation}
where \( \bm{u}_{k,i} \sim \mathcal{N}(\bm{0}, \bm{I}) \). This choice of metric tensor and gradient estimation offers good practical accuracy while significantly reducing the required number of resamplings \( N \) per iteration \( k \).

\section{Numerical Results}
\label{sec:numerical}
To demonstrate the practical applicability of our approach, we conducted numerical simulations of the VQE algorithm using the open-source software \texttt{PennyLane}~\cite{PennyLane}.

\subsection{Example 1: Transverse Field Ising Model}
\label{Example1}

The Transverse Field Ising Model (TFIM) with open boundary conditions is described by the Hamiltonian:

\begin{equation}
    H = J \sum_{i=1}^{N-1} \sigma_i^z \sigma_{i+1}^z + h \sum_{i=1}^{N} \sigma_i^x,
\end{equation}
where \( J \) is the coupling constant, \( h \) is the transverse field strength, and \( \sigma_i^z \), \( \sigma_i^x \) are Pauli matrices acting on site \( i \). The first term represents nearest-neighbor spin interactions along the \( z \)-axis, while the second term introduces quantum fluctuations via the transverse field along the \( x \)-axis.

In this example, we investigate the case with \( J = -1 \) and \( h = -2 \), and approximate the ground state of \( H \) using the hardware-efficient ansatz. This ansatz constructs the wavefunction with a layered quantum circuit that combines parameterized single-qubit rotations (\( R_Y \)) and entangling controlled-NOT (CNOT) gates (see Fig.~\ref{fig:Efficient_Ansatz}).

\begin{figure}[H]
    \centering
    \includegraphics[width=0.5\textwidth]{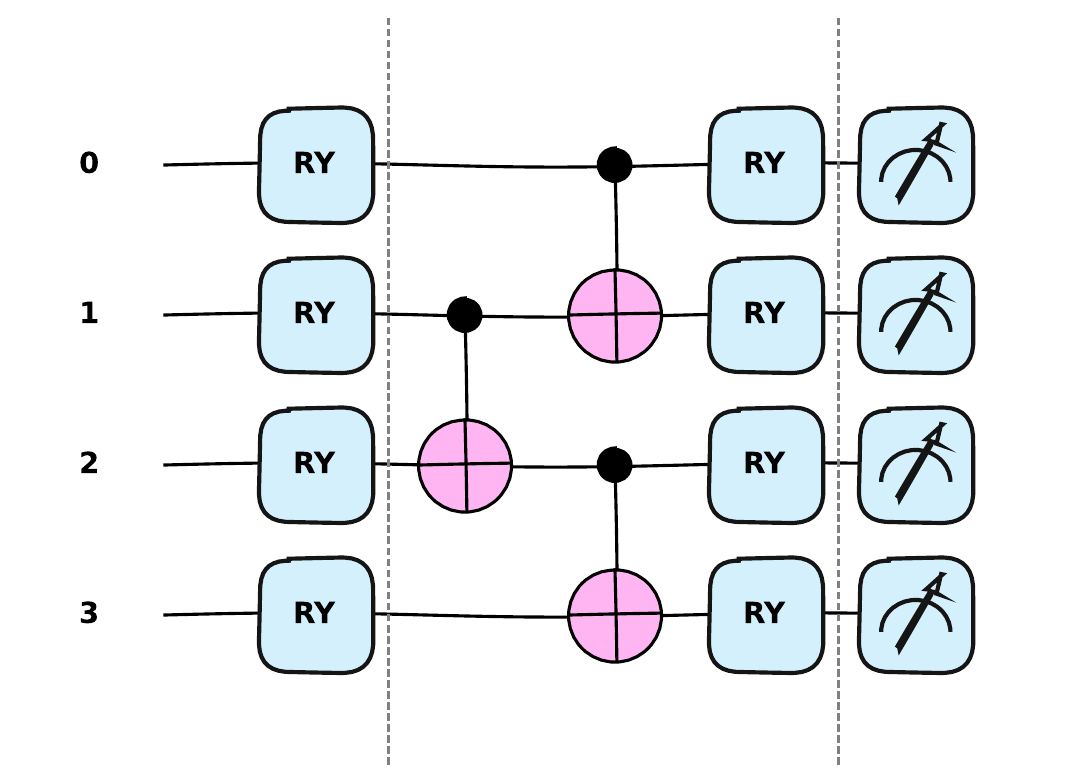}
   \caption{Hardware-efficient ansatz with two layers, using \( R_Y \) rotations and CNOT gates.}
    \label{fig:Efficient_Ansatz}
\end{figure}

In the VQE simulation, we benchmark the following optimizers: Gradient Descent (GD), Quantum Natural Gradient (QNG), Simultaneous Perturbation Stochastic Approximation (SPSA), Quantum Natural SPSA (QNSPSA), Stein Optimizer (using equation \eqref{Gradient_Stein} for the gradient without the Fubini–Study metric), Quantum Natural Stein Optimizer 2 (QNSTEIN2, using the gradient from equation \eqref{Gradient_Stein} and the metric from equation \eqref{Stein_F2}), and Quantum Natural Stein Optimizer 3 (QNSTEIN3, using the gradient from equation \eqref{Gradient_Stein} and the metric from equation \eqref{Stein_F3}). For all stochastic optimizers (SPSA, QNSPSA, STEIN, QNSTEIN2, and QNSTEIN3), resampling is performed at each optimization step \( k \), with \( N = 10 \) samples used for both the gradient and the metric (see \( N \) in \eqref{Gradient_Stein}, \eqref{Stein_F2}, and \eqref{Stein_F3}). The learning rate for all optimizers is fixed at \( 0.01 \). The regularization parameter \( \epsilon \) is set to \( 10^{-2} \) for QNSPSA, QNSTEIN2, and QNSTEIN3, and to \( 10^{-1} \) for QNG. The finite-difference step for SPSA and QNSPSA is fixed at \( 0.05 \), while in QNSTEIN2 and QNSTEIN3, parameters \( b = 2 \) and \( c = 0.05 \) are used. All simulations use 8192 shots, and each optimizer is run for up to 300 steps.

\begin{figure}[H]
    \centering
    \includegraphics[width=\textwidth]{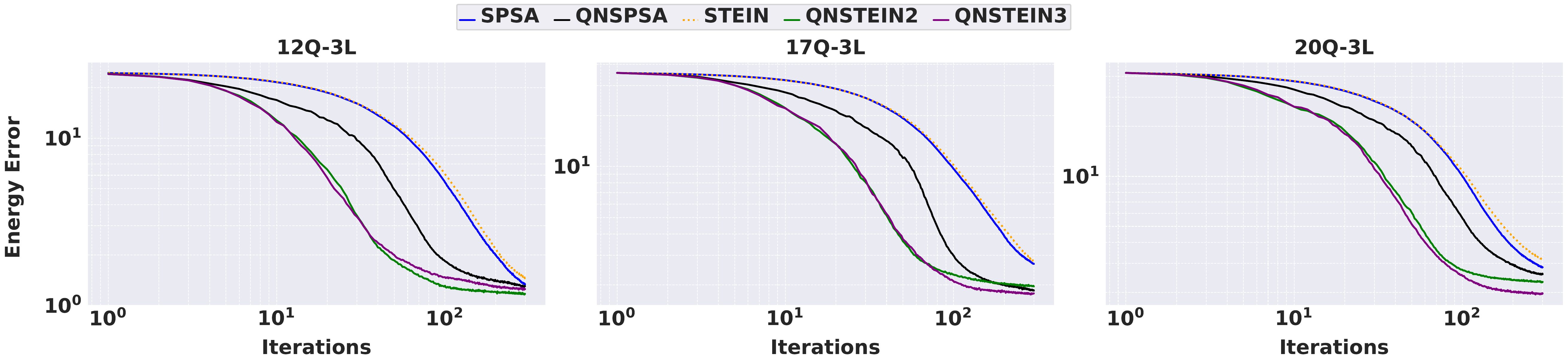}
   \caption{Energy error as a function of iteration steps for different optimization methods—SPSA, QNSPSA, STEIN, QNSTEIN2, and QNSTEIN3—applied to 3-layer circuits with varying numbers of qubits (12, 17, 20). The results are averaged over 30 different random initializations.}
    \label{fig:ising_energy_fixed3L}
\end{figure}
\begin{figure}[H]
    \centering
    \includegraphics[width=\textwidth]{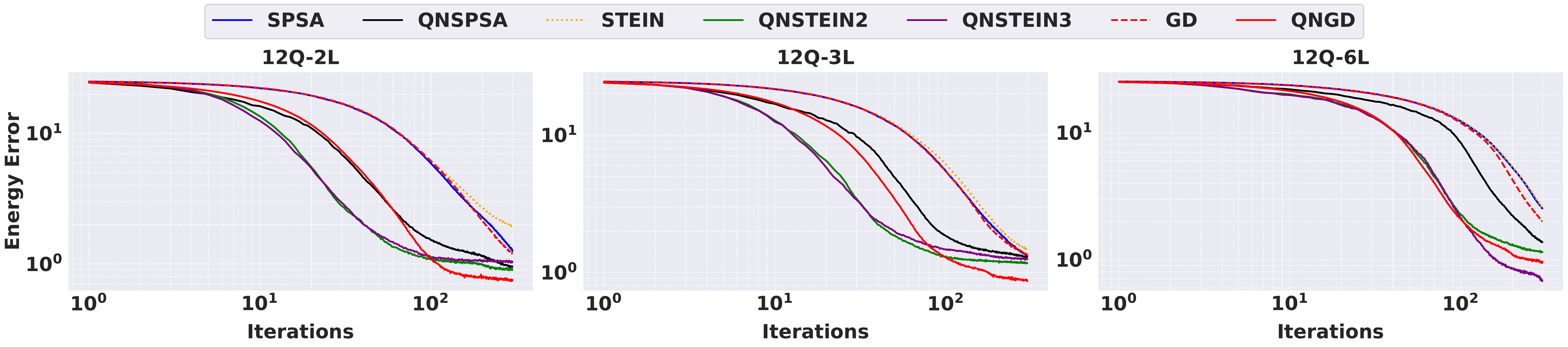}
    \caption{Same as in Figure~\ref{fig:ising_energy_fixed3L}, but this time with 12 qubits fixed and varying layers (\( L = 2, 3, 6 \)).}
    \label{fig:ising_energy_fixed12Q}
\end{figure}
\begin{figure}[H]
    \centering
    \includegraphics[width=\textwidth]{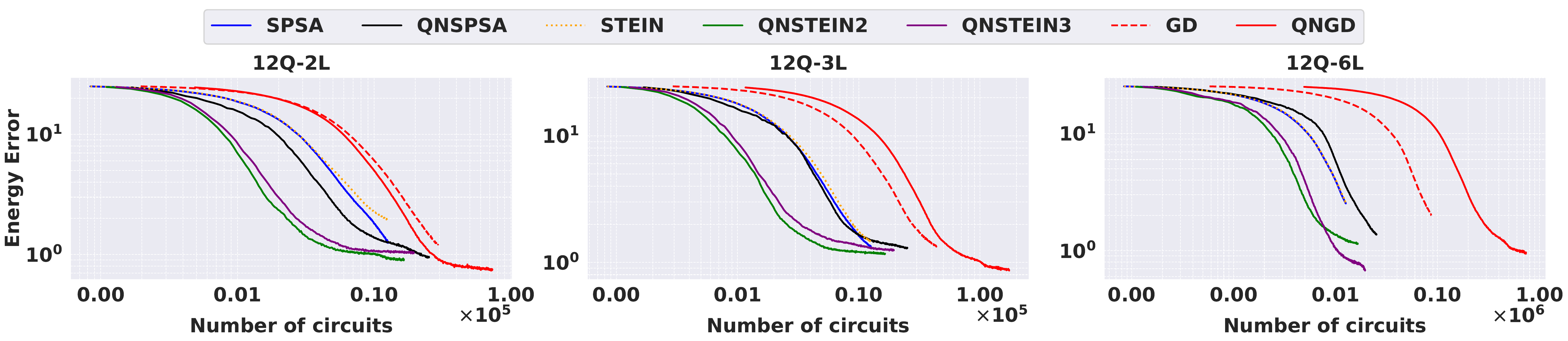}
    \caption{Same as in Figure~\ref{fig:ising_energy_fixed12Q}, but here we plot the energy error as a function of the number of quantum circuits required for convergence.}
    \label{fig:ising_energy_fixed12Q_Ciruit}
\end{figure}

The performance of the optimizers is highlighted in Figures~\ref{fig:ising_energy_fixed3L}–\ref{fig:ising_energy_fixed12Q_Ciruit}. Figure~\ref{fig:ising_energy_fixed3L} shows the energy error (i.e., the difference between the obtained and target energies) as a function of iteration for a fixed three-layer circuit with varying numbers of qubits, indicating that the QNSTEINs tend to exhibit an overall better convergence rate in these benchmarks. Figure~\ref{fig:ising_energy_fixed12Q} supports this observation for a fixed 12-qubit system with varying numbers of layers. Figure~\ref{fig:ising_energy_fixed12Q_Ciruit} illustrates the quantum resources required for convergence, showing that the QNSTEINs achieve target energies with the fewest circuit evaluations, thereby reducing quantum resource usage. Across all simulations, there is no significant difference between SPSA and STEIN in terms of convergence behavior and quantum resource consumption.

\subsection{Example 2: Schwinger Model}

The Schwinger Model \cite{Schwinger} is a (1+1)-dimensional quantum field theory that serves as a fundamental framework for studying quantum electrodynamics (QED) in reduced dimensions. Although it is lower-dimensional, the model retains essential features of more complex gauge theories, such as confinement, charge screening, and chiral symmetry breaking. These characteristics make it an important system for investigating non-perturbative effects in quantum field theory. Additionally, its lattice formulation facilitates efficient mapping onto quantum hardware, making it a promising candidate for exploring quantum simulations of gauge theories.

The dynamics of the Schwinger Model on a lattice are described by the Kogut-Susskind Hamiltonian. After mapping the Hamiltonian to qubits using the Jordan-Wigner transformation, the Hamiltonian for our quantum computing optimization task is:
\begin{equation}
H = \frac{x}{2} \sum_{n=0}^{N-2} \left( \sigma_n^x \sigma_{n+1}^x + \sigma_n^y \sigma_{n+1}^y \right) + \frac{\mu}{2} \sum_{n=0}^{N-1} \left[1 + (-1)^n \sigma^z_n\right] + \sum_{n=0}^{N-2} \left(l + \frac{1}{2} \sum_{k=0}^n (-1)^k \sigma^z_k\right)^2,
\label{eq:H_schwinger}
\end{equation}
where the operators \(\sigma_n^x, \sigma_n^y, \sigma_n^z\) represent the Pauli matrices applied to the qubit at site \(n\). The parameter \(x = 1/(g^2 a^2)\) is related to the coupling constant \(g\) and lattice spacing \(a\), while \(\mu = 2m/g^2 a\) is the dimensionless fermion mass term, with \(m\) being the fermion mass. The parameter \(l\) is a background electric field contribution, associated with the zero-mode of the gauge field.

To find the ground state of the Hamiltonian \eqref{eq:H_schwinger} using VQE, we employ the ansatz shown in Figure~\ref{fig:ansatz_schw} \cite{Yibin}, with parameters set to \( l = 0 \), \( x = 1 \), and \( \mu = 0.5 \). We benchmark the same optimizers and use the same hyperparameters as in Example 1; however, in this case, all optimizers run for up to 200 steps, with \( N = 15 \) samples and shot noise set to 10024.
 
\begin{figure}[H]
    \centering
    \begin{minipage}[t]{0.2\textwidth} 
        \centering
        \makebox[\textwidth]{(a)} \\ 
        \includegraphics[width=\textwidth]{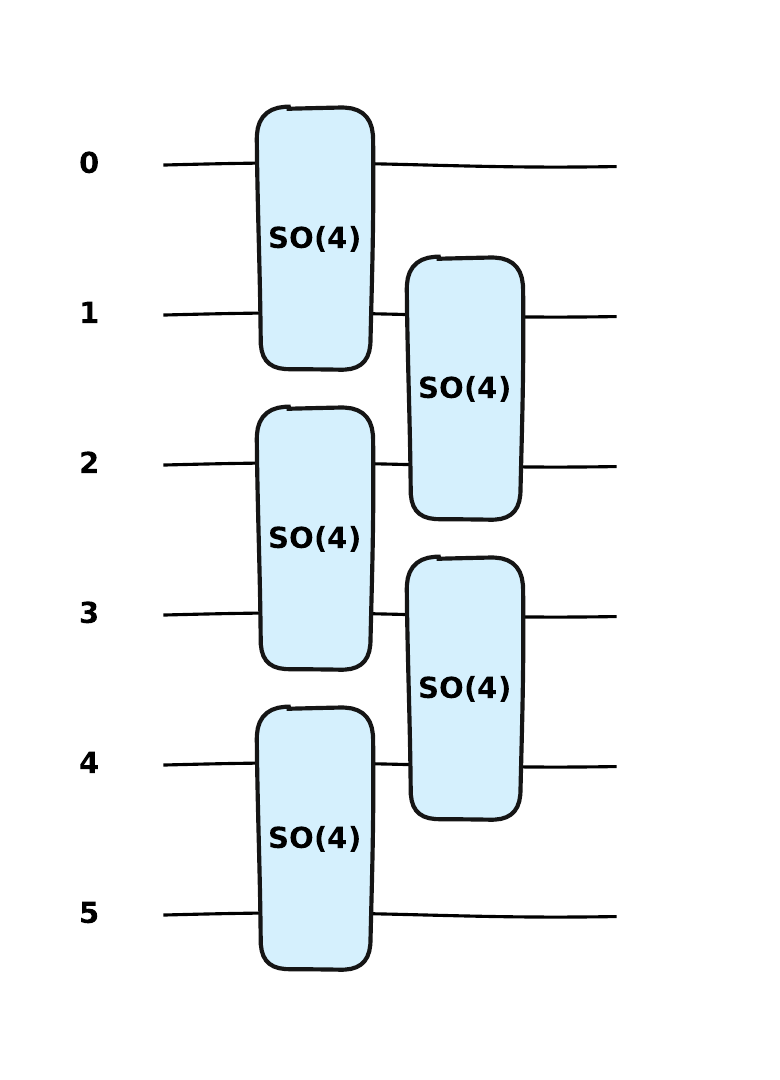} 
    \end{minipage}%
    \hfill
    \begin{minipage}[t]{0.65\textwidth} 
        \centering
        \makebox[\textwidth]{(b)} \\ 
        \[
  \begin{tabular}{@{}c@{\hspace{0.05cm}=\hspace{0.5cm}}c@{}}
    \raisebox{-0.5\height}{%
      \includegraphics[
        height=2cm,
        trim=0 10pt 0 10pt,
        clip
      ]{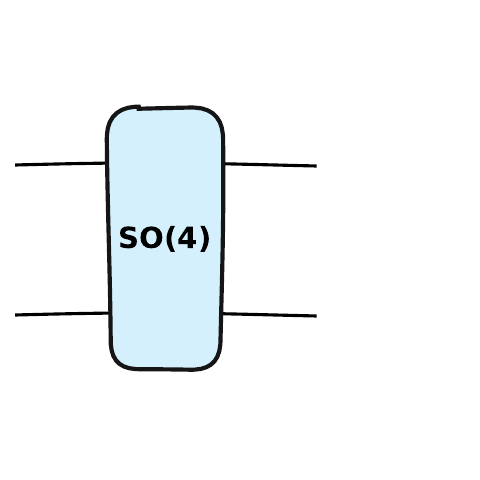}%
    }
    &
    \raisebox{-0.5\height}{%
      \includegraphics[
        height=2cm,
        trim=0 10pt 0 10pt,
        clip
      ]{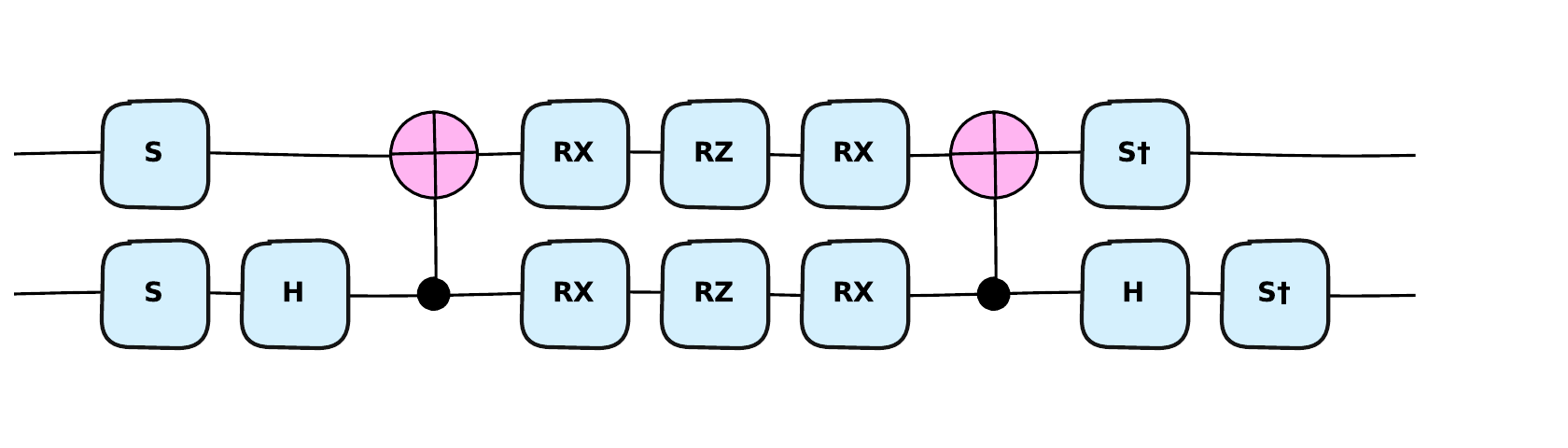}%
    }
  \end{tabular}
\]
    \end{minipage}
    \caption{(a) The ansatz used in the VQE algorithm to approximate the ground state of the Schwinger Hamiltonian, shown here with a single layer incorporating universal $SO(4)$ gates. (b) The decomposition of the two-qubit $SO(4)$ gate into single-qubit phase gates $S$ and its conjugate transpose $S^\dagger$, the Hadamard gate $H$, and rotation gates ($RX$, $RZ$), along with entangling two-qubit CNOT operations.}
    \label{fig:ansatz_schw}
\end{figure}
\begin{figure}[H]
    \centering
    \includegraphics[width=\textwidth]{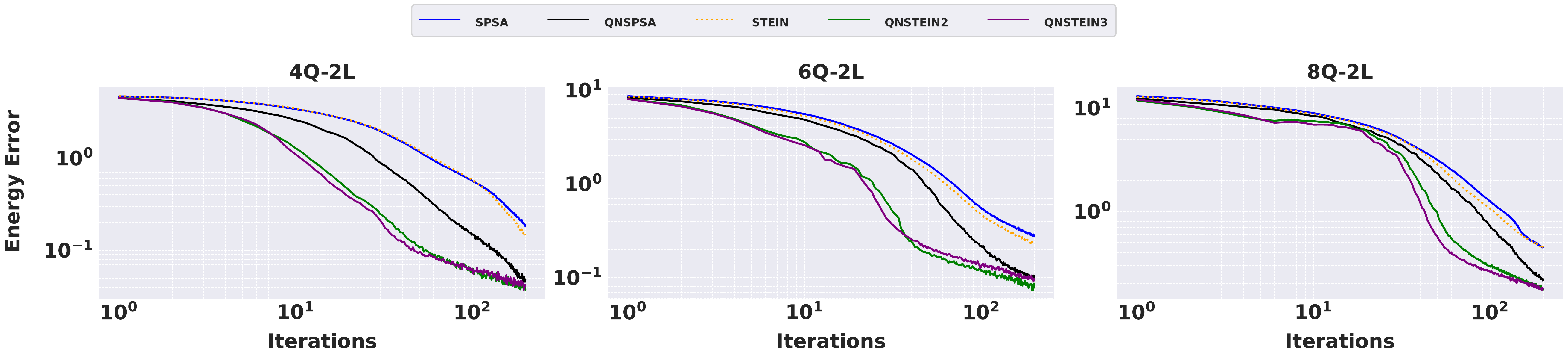}
    \caption{Energy error as a function of iteration steps for different optimization methods—SPSA, QNSPSA, STEIN, QNSTEIN2, and QNSTEIN3—applied to 2-layer circuits with varying numbers of qubits (4, 6, 8). The Results are averaged over 30 random initializations of the variational parameters.}
    \label{fig:schw_energy_fixed2L}
\end{figure}
\begin{figure}[H]
    \centering
    \includegraphics[width=\textwidth]{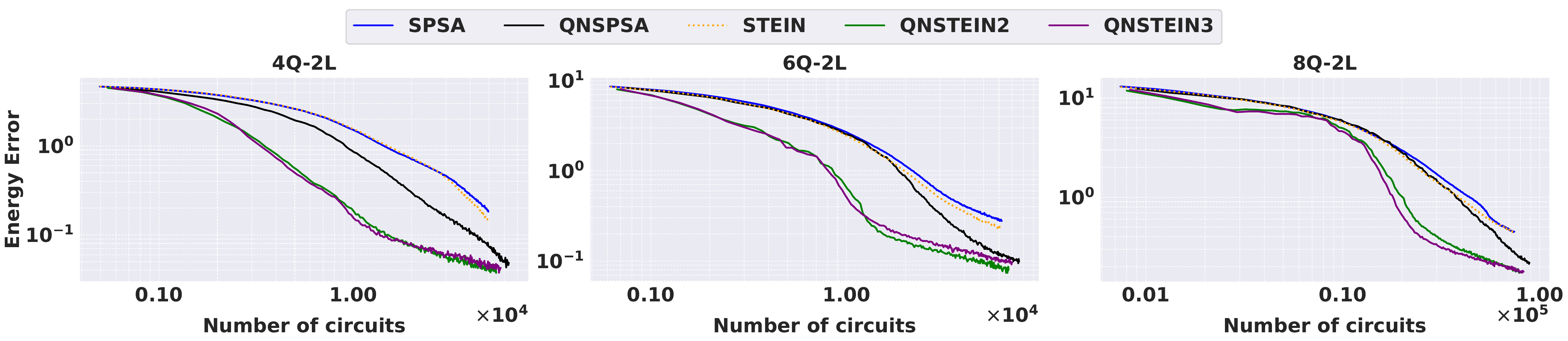}
    \caption{Same as in Figure~\ref{fig:schw_energy_fixed2L}, but here we plot the energy error as a function of the number of quantum circuits required for convergence.}
    \label{fig:schw_energy_fix2L_circuits}
\end{figure}

As in the TFIM Hamiltonian, the results in Figures~\ref{fig:schw_energy_fixed2L} and~\ref{fig:schw_energy_fix2L_circuits} show that the QNSTEIN optimizers also perform well in terms of convergence and in reducing the quantum resources required for the more complex, higher-energy physics Hamiltonian of the Schwinger model.
\section{Conclusions and Outlook}
\label{sec:conclusion}

In this work, we proposed an approach for estimating the Quantum Fisher Information Matrix (QFIM) within Variational Quantum Algorithms (VQAs) using Stein’s identity. The method significantly reduces computational complexity from \( O(m^2) \) to a constant, effectively preserving parameter correlations and enhancing performance in quantum optimization tasks. Compared to the existing QFIM estimation method based on Simultaneous Perturbation Stochastic Approximation (SPSA) proposed in \cite{Gacon}, which also reduces the computational complexity to a constant, our Stein-based estimator requires fewer quantum circuit evaluations (function overlaps), thereby reducing the overall quantum resource requirements.

We validated our approach by benchmarking various optimizers in VQE simulations applied to the Transverse Field Ising Model and the lattice Schwinger Model. The results suggest that Quantum Natural Stein algorithms tend to achieve lower energy values with fewer iterations and reduced quantum resource consumption in these benchmarks, highlighting the potential benefits of integrating our Stein-based QFIM estimator into practical variational quantum algorithms.

Quantitatively, our estimators notably reduce the number of circuit evaluations per optimization step compared to Quantum Natural SPSA (QNSPSA). Specifically, QNSPSA typically requires two evaluations for gradients and four for the Fubini–Study metric, totaling six evaluations per step. In contrast, the two-evaluation Stein estimator (QNSTEIN2) uses two circuits for gradient and two for metric estimation, totaling four evaluations per step—a 1.5× improvement. The three-evaluation Stein estimator (QNSTEIN3) employs two circuits for gradient estimation and three for metric estimation, totaling five circuits per step—a 1.2× improvement. Considering metric evaluations alone, QNSTEIN2 achieves a 2× reduction, and QNSTEIN3 achieves approximately a 1.33× reduction. Although practical factors such as finite-shot noise, estimator variance, and parameter settings (\(c\) and \(b\)) may slightly affect these theoretical gains, appropriate tuning or modest increases in sampling generally allow the theoretical advantages to translate into practical benefits.

The presented framework is flexible and opens avenues for further developments, including exploring alternative perturbation distributions and adaptive resampling strategies to enhance accuracy and robustness. Moreover, this QFIM estimation technique may naturally extend to other quantum computational domains, such as quantum metrology.

\section*{Acknowledgements}

This work was supported by the Ministry of Science, Research, and Culture of the State of Brandenburg through the Centre for Quantum Technologies and Applications (CQTA) at DESY (Germany) and by the German Ministry of Education and Research (BMBF) under project NiQ (Grant No. 13N16203). The author thanks Yibin Guo (DESY, Germany) for valuable discussions on the Schwinger model.

\begin{figure}[H]
\centering
\includegraphics[width=0.2\textwidth]{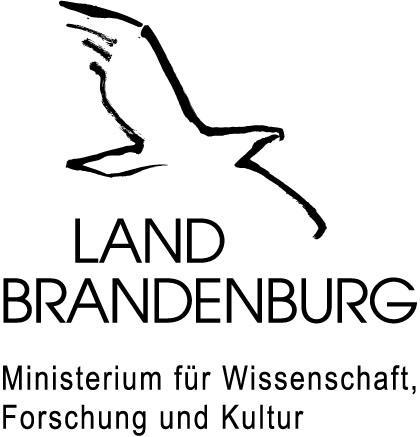}
\end{figure}

\appendix
\section{Fubini-Study Metric} \label{app:Fubini}

To derive the Fubini-Study metric tensor \(F_{ij}(\bm{\theta})\), we investigate the infinitesimal distance between adjacent quantum states in the parameter space. This metric captures the local differential geometry of the quantum state manifold, reflecting how the quantum state \(|\psi(\bm{\theta})\rangle\) responds to infinitesimal parameter variations~\cite{Stokes2020, Meyer2021, Liu2020, S.Sidhu2020}.

We begin with the normalization condition of the quantum state:
\begin{equation}
\langle \psi(\bm{\theta}) | \psi(\bm{\theta}) \rangle = 1.
\end{equation}
Differentiating this condition with respect to \(\theta^i\) yields:
\begin{equation}
\left\langle \psi(\bm{\theta})| \partial_i \psi(\bm{\theta}) \right\rangle 
+ \left\langle \partial_i \psi(\bm{\theta})| \psi(\bm{\theta}) \right\rangle = 0.
\label{eq:33}
\end{equation}
Taking a further derivative with respect to \(\theta^j\) leads to:
\begin{equation}
\left\langle \psi(\bm{\theta})| \partial_i \partial_j \psi(\bm{\theta}) \right\rangle 
+ \left\langle \partial_i \partial_j \psi(\bm{\theta})| \psi(\bm{\theta}) \right\rangle 
+ \left\langle \partial_i \psi(\bm{\theta})| \partial_j \psi(\bm{\theta}) \right\rangle 
+ \left\langle \partial_j \psi(\bm{\theta})| \partial_i \psi(\bm{\theta}) \right\rangle = 0.
\label{eq:34}
\end{equation}
We now consider the perturbed quantum state \(|\psi(\bm{\theta} + \delta \bm{\theta})\rangle\), where \(\delta \bm{\theta}\) is a small displacement vector. Expanding this state in a Taylor series yields:
\begin{equation}
|\psi(\bm{\theta} + \delta\theta)\rangle = |\psi(\bm{\theta})\rangle + \partial_i |\psi(\bm{\theta})\rangle \delta\theta^i + \frac{1}{2} \partial_i \partial_j |\psi(\bm{\theta})\rangle \delta\theta^i \delta\theta^j + \mathcal{O}(\delta\theta^3).
\end{equation}
Taking the inner product with the original state results in:
\begin{equation}
\langle \psi(\bm{\theta}) | \psi(\bm{\theta} + \delta \bm{\theta}) \rangle = 1 + \langle \psi(\bm{\theta}) | \partial_i \psi(\bm{\theta}) \rangle \delta\theta^i + \frac{1}{2} \langle \psi(\bm{\theta}) | \partial_i \partial_j \psi(\bm{\theta}) \rangle \delta\theta^i \delta\theta^j + \mathcal{O}(\delta\theta^3).
\end{equation}
The fidelity between these quantum states, accurate to second order in \(\delta \bm{\theta}\), is given by:
\begin{align}
|\langle \psi(\bm{\theta})| \psi(\bm{\theta} + \delta \bm{\theta}) \rangle|^2 
&= 1 + \left[ \left\langle \psi(\bm{\theta})| \partial_i \psi(\bm{\theta}) \right\rangle + \left\langle \partial_i \psi(\bm{\theta})| \psi(\bm{\theta}) \right\rangle \right] \delta\theta^i \nonumber \\
&\quad + \left[ \left\langle \partial_i \psi(\bm{\theta})|  \psi(\bm{\theta}) \right\rangle  \left\langle \psi(\bm{\theta}) | \partial_j \psi(\bm{\theta}) \right\rangle \right] \delta\theta^i \delta\theta^j \nonumber \\
&\quad + \frac{1}{2} \left[ \left\langle \psi(\bm{\theta}) | \partial_i \partial_j \psi(\bm{\theta}) \right\rangle + \left\langle \partial_i \partial_j \psi(\bm{\theta}) | \psi(\bm{\theta}) \right\rangle \right] \delta\theta^i \delta\theta^j + \cdots
\end{align}
By incorporating Eqs.~\eqref{eq:33} and \eqref{eq:34}, this simplifies to:
\begin{equation}
\begin{aligned}
\bigl|\langle \psi(\bm{\theta}) \mid \psi(\bm{\theta} + \delta\bm{\theta}) \rangle\bigr|^2
&= 1
  + \Bigl[
      \langle \partial_i \psi(\bm{\theta}) \mid \psi(\bm{\theta}) \rangle\,
      \langle \psi(\bm{\theta}) \mid \partial_j \psi(\bm{\theta}) \rangle \\[6pt]
&\quad
    - \tfrac12 \,\bigl(
        \langle \partial_i \psi(\bm{\theta}) \mid \partial_j \psi(\bm{\theta}) \rangle
        +\,
        \langle \partial_j \psi(\bm{\theta}) \mid \partial_i \psi(\bm{\theta}) \rangle
      \bigr)
  \Bigr] \,\delta\theta^i\,\delta\theta^j
  + \cdots
\end{aligned}
\end{equation}
Using the small-angle approximation for the Fubini-Study distance between projective rays:
\begin{equation}
d^2(P_\psi, P_\phi) = \arccos^2(|\langle \psi | \phi \rangle|) \approx 1 - |\langle \psi | \phi \rangle|^2 + \mathcal{O}((1 - |\langle \psi | \phi \rangle|^2)^2),
\end{equation}
we obtain the following expression for the infinitesimal squared distance:
\begin{equation}\label{eq:qgt_expanded}
\begin{aligned}
d^2\bigl(P_{\psi(\bm{\theta})},\,P_{\psi(\bm{\theta}+\delta\bm{\theta})}\bigr)
&= \Bigl[
     \tfrac12\,\bigl(
       \langle \partial_i \psi(\bm{\theta}) \mid \partial_j \psi(\bm{\theta})\rangle
       + \langle \partial_j \psi(\bm{\theta}) \mid \partial_i \psi(\bm{\theta})\rangle
     \bigr) \\
&\quad\; -\,\langle \partial_i \psi(\bm{\theta}) \mid \psi(\bm{\theta})\rangle\,
           \langle \psi(\bm{\theta}) \mid \partial_j \psi(\bm{\theta})\rangle
  \Bigr]\,d\theta^i\,d\theta^j \,.
\end{aligned}
\end{equation}
The first term is evidently real, and the second is also real due to the normalization constraint, which implies:
\begin{equation}
\text{Re}\left[ \left\langle \psi(\bm{\theta})| \partial_i \psi(\bm{\theta}) \right\rangle \right] = 0.
\end{equation}
Thus, the infinitesimal squared distance corresponds to the real part of the quantum geometric tensor:
\begin{align}
d^2(P_{\psi(\bm{\theta})}, P_{\psi(\bm{\theta} + \delta \bm{\theta})}) 
&= \text{Re} \left[
\left\langle \partial_i \psi(\bm{\theta})| \partial_j \psi(\bm{\theta}) \right\rangle 
- \left\langle \partial_i \psi(\bm{\theta})| \psi(\bm{\theta}) \right\rangle 
\left\langle \psi(\bm{\theta})| \partial_j \psi(\bm{\theta}) \right\rangle 
\right] d \theta^i d \theta^j. \label{eq:fs_distance}
\end{align}
This leads directly to the definition of the Fubini-Study metric tensor:
\begin{equation}
F_{ij}(\bm{\theta}) = \text{Re} \left[
\left\langle \partial_i \psi(\bm{\theta}) | \partial_j \psi(\bm{\theta}) \right\rangle 
- \left\langle \partial_i \psi(\bm{\theta}) | \psi(\bm{\theta}) \right\rangle 
\left\langle \psi(\bm{\theta}) | \partial_j \psi(\bm{\theta}) \right\rangle 
\right]. \label{eq:fs_metric}
\end{equation}
One of the most costly traditional methods for estimating the Fubini–Study metric, with a computational complexity of \( O(d^2) \), is the parameter-shift rule \cite{Mari}
\begin{equation}
\begin{aligned}
F_{j_1, j_2}(\bm{\theta}) = \frac{1}{4} \Big[ &|\langle \psi(\bm{\theta}) | \psi(\bm{\theta} + (\bm{e}_{j_1} + \bm{e}_{j_2})\pi/2) \rangle|^2 \\
&- |\langle \psi(\bm{\theta}) | \psi(\bm{\theta} + (\bm{e}_{j_1} - \bm{e}_{j_2})\pi/2) \rangle|^2 \\
&- |\langle \psi(\bm{\theta}) | \psi(\bm{\theta} + (-\bm{e}_{j_1} + \bm{e}_{j_2})\pi/2) \rangle|^2 \\
&+ |\langle \psi(\bm{\theta}) | \psi(\bm{\theta} - (\bm{e}_{j_1} + \bm{e}_{j_2})\pi/2) \rangle|^2 \Big].
\end{aligned}
\end{equation}
Here, \( \bm{e}_{j} \) denotes the unit vector along the \( \theta_{j} \) axis.
\section{Tensor‐averaging error}
\label{Tensor_error}

\begin{remark}[Tensor‐averaging error]
Strictly speaking, the update of \(\bar{\mathbf{F}}_k\) in equation~\eqref{eq:QNG_updat} mixes two Riemannian metric tensors that lie in different tangent spaces at \(\boldsymbol{\theta}_{k-1}\) and \(\boldsymbol{\theta}_k\). A fully geometric treatment would first \emph{pull back}
\(\mathbf{F}(\boldsymbol{\theta}_{k-1})\)
along the map \(\boldsymbol{\theta}_{k-1} \to \boldsymbol{\theta}_k\) into the tangent space at \(\boldsymbol{\theta}_k\), and only then average the two metrics in the same space.

However, for a small learning rate \(\eta\), it can be shown that
\begin{equation}
  \bigl\|\boldsymbol{\theta}_k - \boldsymbol{\theta}_{k-1}\bigr\|
  = \bigl\|\eta\,\mathbf{F}(\boldsymbol{\theta}_{k-1})^{-1}\,\nabla\mathcal{L}(\boldsymbol{\theta}_{k-1})\bigr\|
  = O(\eta),
\end{equation}
and if \(\mathbf{F}(\boldsymbol{\theta})\) is Lipschitz continuous with constant \(L\), then
\begin{equation}
  \bigl\|\mathbf{F}(\boldsymbol{\theta}_k) - \mathbf{F}(\boldsymbol{\theta}_{k-1})\bigr\|
  \le L\,\bigl\|\boldsymbol{\theta}_k - \boldsymbol{\theta}_{k-1}\bigr\|
  = O(\eta).
\end{equation}
Since the pullback operator itself deviates from the identity by \(O(\eta)\) over such a small step, the combined geometric bias—from neglecting the pullback and averaging two nearby metrics—is \(O(\eta)\) per step. For typical VQE learning rates (\(\eta \ll 1\)), this remains negligible compared to stochastic and shot-noise fluctuations.
\label{rem_tensor}
\end{remark}

\section{Additional Experiments and Figures}

In Figure~\ref{fig:ising_N}, we fix the resampling size \( N = 5 \) for QNSTEIN2 and QNSTEIN3, while varying the resampling size \( N \) for QNSPSA. We observe that QNSPSA requires a larger resampling size (greater than 5) to achieve a convergence rate comparable to the QNSTEIN optimizers. However, this increase in \( N \) significantly raises the quantum resources needed for convergence.

\begin{figure}[H]
    \centering
    \includegraphics[width=\textwidth]{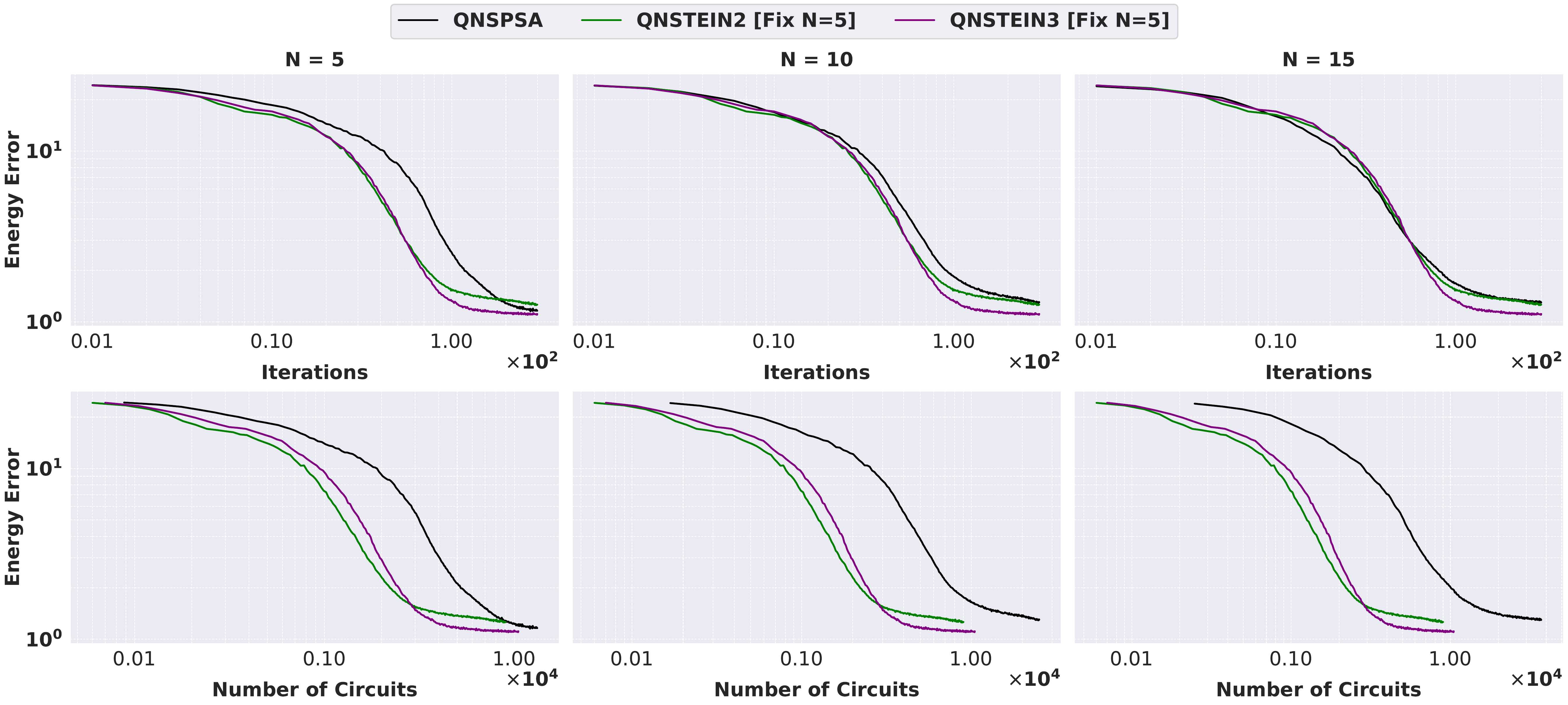}
    \caption{The resampling size \( N = 5 \) is fixed for QNSTEIN2 and QNSTEIN3, while \( N \) is varied for QNSPSA. The experimental conditions are the same as in Section~\ref{Example1} (12 qubits and 3 layers), and results are averaged over 30 different random initializations.}
    \label{fig:ising_N}
\end{figure}

\end{document}